\documentclass[a4paper,12pt]{article}
\usepackage{latexsym,amsmath,amsfonts,amsbsy,amscd,amsxtra,amsgen,amsopn,bbm,amsthm,amssymb}

\begin{document}
\title{\textbf{An absolute quantum energy inequality for the Dirac field in curved spacetime}}
\author{Calvin J. Smith\footnote{Electronic address: {\tt calvin.smith@ucd.ie}} \\ \textit{School of Mathematical Sciences, University College Dublin,} \\ \textit{Belfield, Dublin 4, Ireland}}
\date{15 May 2007}

\maketitle
\newtheorem{theorem}{Theorem}[section]
\newtheorem*{defn}{Definition}
\newtheorem{lemma}[theorem]{Lemma}
\newtheorem{corollary}[theorem]{Corollary}
\newtheorem{claim}[theorem]{Claim}
\newtheorem{prop}[theorem]{Proposition}
\newtheorem{proposal}[theorem]{Proposal}
\newtheorem{example}[theorem]{Example}
\newtheorem{summary}[theorem]{Summary}
\newtheorem{axiom}{Wald Axiom}

\newcommand{\supp}{\textrm{supp }}
\newcommand{\singsupp}{\textrm{singsupp }}
\newcommand{\QED}{\qed \\}
\newcommand{\R}{\mathbb{R}}
\newcommand{\Z}{\mathbb{Z}}
\newcommand{\C}{\mathbb{C}}
\newcommand{\N}{\mathbb{N}}
\newcommand{\M}{\mathcal{M}}
\newcommand{\D}{\mathcal{D}}
\newcommand{\U}{\mathcal{U}}
\newcommand{\E}{\mathcal{E}}
\newcommand{\s}{\mathcal{S}}
\newcommand{\open}{\mathcal{O}}
\newcommand{\mink}{\mathbb{M}}
\newcommand{\norm}{\parallel}
\newcommand{\hilbert}{\mathcal{H}}
\newcommand{\vol}{\mathrm{vol}}
\newcommand{\diag}{\mathrm{diag\,}}
\newcommand{\Char}{\mathrm{Char \,}}
\newcommand{\LIM}{\mathrm{l.i.m.}}
\newcommand{\loc}{\mathrm{loc}}
\newcommand{\nnnl}{\nonumber \\}
\newcommand{\ren}{\mathrm{ren}}
\newcommand{\normal}{\mathcal{N}}
\newcommand{\GL}{\mathrm{GL}_n(\C)}
\newcommand{\dvol}{\mathrm{dvol}}
\newcommand{\rmd}{\mathrm{d}}
\newcommand{\tr}{\mathrm{tr }}

\newcommand{\Cauchy}{\mathcal{C}}
\newcommand{\mbk}{\mathbf{k}}
\newcommand{\fin}{\mathrm{fin}}
\newcommand{\dirac}{\nabla\hspace{-3mm}{/}}
\newcommand{\W}{\mathrm{W}}
\newcommand{\cW}{\mathcal{W}}
\newcommand{\had}{\mathrm{H}}
\newcommand{\spin}{\mathrm{sp}}
\newcommand{\cosp}{\mathrm{cosp}}
\newcommand{\Spin}{\mathrm{Spin }}
\newcommand{\J}{\mathrm{J}}
\newcommand{\IN}{\mathrm{in}}
\newcommand{\OUT}{\mathrm{out}}
\newcommand{\vin}{\Omega^\IN}
\newcommand{\SM}{S\M}
\newcommand{\FM}{F\M}
\newcommand{\DM}{D\M}
\newcommand{\T}{\mathrm{t}}

\begin{abstract}
Quantum Weak Energy Inequalities (QWEIs) are results which limit
the extent to which the smeared renormalised energy density of a quantum field can be negative. On globally hyperbolic spacetimes the massive quantum Dirac field is known to obey a QWEI in terms of a reference state chosen arbitrarily from the class of Hadamard states; however, there exist spacetimes of interest on which state-dependent
bounds cannot be evaluated. In this paper we prove the first QWEI for the
massive quantum Dirac field on four dimensional
globally hyperbolic spacetime in which the bound
depends only on the local geometry; such a QWEI is known as an
absolute QWEI.
\end{abstract}

\section{Introduction}

When formulating the classical theory of general relativity it is necessary to impose
certain energy conditions on the source matter fields being considered.
The most commonplace of these energy conditions is the \textit{weak energy condition}, $T_{ab}k^ak^b \geq
0$ for every timelike vector field $k$, which entails that observers only
encounter positive energy densities. However, it has been known since 1965 that, unlike most classical physics models,
no (Wightman) quantum field theory can obey pointwise energy
conditions \cite{epstein}. Moreover, it is possible to show
that the negative energy density arising from a quantum field
theoretic source is unbounded (from below) in magnitude \cite{fewster lisbon}. This startling feature of quantum field theory is often used, in the context of the semi-classical Einstein equation $G_{ab}=8\pi G\langle T^\ren_{ab}\rangle_\omega$, to support so-called `designer spacetimes' like Alcubierre's warp drive \cite{alcubierre} or traversable worm hole geometries. Following Ford's
\cite{ford} observation that it is possible to bound the magnitude and
duration of the flux of negative energy of a quantum field source, work
began in earnest to prove that the averaged (expectation value of
the) energy density of a quantum field was bounded from below. A
suitable definition, sufficient for our purposes\footnote{A more
general, and rigorous, definition of a quantum energy inequality is
given in \cite{fewster paris}.}, is that a
worldline \textit{quantum weak energy inequality} (QWEI) is a result of the form
\begin{equation}\label{wqwei}
\int_\R \rmd \tau \, \langle \rho^\mathrm{ren} \rangle_\omega (\gamma(\tau)) F(\tau) \geq
- \mathcal{B} > -\infty
\end{equation}
where $F$ is some appropriately chosen sampling function, $\gamma:\R
\mapsto \M$ is a timelike worldline and $\langle
\rho^\mathrm{ren}\rangle_\omega$ is the (expectation value of the) energy density of the quantum
field in a state $\omega$. In this discussion we shall exclusively
consider the massive quantum Dirac field in a smooth four-dimensional globally hyperbolic
spacetime $(\M,g)$. Moreover, we shall only consider the Hadamard states
of the Dirac field as this is a sufficient class of states to
renormalise the stress energy density. 

Typically, the bound $\mathcal{B}$ featuring
in (\ref{wqwei}) is a function of another state of the
theory usually called a reference state; these QWEIs are known as \textit{difference QWEIs}. Due to the
work of Fewster and his collaborators \cite{dawson fewster, fewster verch}, difference
QWEIs are known in great generality for the Dirac field in curved spacetime. (For a brief review of QWEIs for other fields the reader is directed to section one of \cite{fewster smith} and the references therein). Difference
QWEIs have been instrumental in constraining the likelihood of designer
spacetime manifestation; however, there exist spacetimes on
which one does not know how to write down the closed form expressions
for states necessary for the evaluation of the difference QWEI bound.
Indeed, the warp drive is an example of a spacetime on which it is not
currently known how to obtain explicit expressions for Hadamard
states. Therefore, it is desirable to have a lower bound $\mathcal{B}$ which is
state independent and constructed only from the local geometry; such a
bound is known as an \textit{absolute QWEI}. Currently, for the Dirac
field, absolute QWEIs are known only for the conformally invariant
\cite{vollick} and massive field in two-dimensions \cite{dawson} and the massive field in
four-dimensional flat spacetime \cite{fewster mistry}. In this
discussion we state and prove the first absolute QWEI for the massive Dirac
field in arbitrary four-dimensional globally hyperbolic spacetime.  The argument is an adaptation of Fewster's earlier work with Verch \cite{fewster verch} and Dawson \cite{dawson fewster} and is to be viewed as a companion to the analogous result for the Klein-Gordon field \cite{fewster smith}.

Our result may be stated as follows: Let $(\M,g)$ be a classical curved four-dimensional spacetime.
Here $\mathcal{M}$ is a four-dimensional smooth manifold (assumed Hausdorff,
paracompact and without boundary) with a Lorentz metric $g_{ab}$ of signature
($+---$). Furthermore, we require $(\M,g)$ to be \textit{globally hyperbolic}, that is
$\M$ contains a Cauchy surface. In addition we assume that an orientation, time orientation and spin structure have been chosen. It may be shown that on such a background one may formulate the quantum Dirac field and a notion of Hadamard states. The essential feature of Hadamard states is that they all share a common singularity structure; in particular their two-point functions, and their Dirac adjoints, have a local and covariantly determined singular expansion. We denote the Hadamard series corresponding to the singularity structure of the Dirac two-point function $\cW_\omega$ by ${\,}^\psi H_k^{(+)}$ and that corresponding to the singular structure of the adjoint Dirac two-point function $\cW^\Gamma_\omega$ by ${\,}^\psi H^{(-)}_k$. The salient feature of such states is that one may define a finite stress energy density $\langle \rho^\fin \rangle_\omega(x)$ by using $\cW_\omega-{\,}^\psi H_k^{(+)}$ and a point-splitting prescription. Indeed, define an operator $\rho^\mathrm{split}$ such that the finite contribution to the energy density is given by $\rho^\textrm{split}$ acting on the regularised two-point function of the Dirac field $\cW_\omega - {\,}^\psi H_k^{(+)}$; i.e. $$\langle \rho^\fin \rangle_\omega (x) := \lim_{x'\rightarrow x} [\rho^\mathrm{split} ( \cW_\omega -{\,}^\psi H_1^{(+)} ) ](x,x') \, .$$  The precise form of $\rho^\textrm{split}$ is given in \S\ref{energy section}. The quantity $\langle \rho^\fin\rangle_\omega$ is equal to the renormalised energy density modulo a local curvature term; we shall return to this issue later after proving our main result in theorem \ref{main result}. Our result then reads for any real valued $f\in C^\infty_0(\R)$ and Hadamard state $\omega$
\begin{equation}
\int_\R \rmd \tau \, \langle \rho^\mathrm{ren} \rangle_\omega (\gamma(\tau)) f^2(\tau) \geq
- \mathcal{B} \, ,
\end{equation}
modulo local curvature terms, where $\mathcal{B}$ is of the form 
\begin{eqnarray}
\mathcal{B} &=& \int_{\R^+}\frac{\rmd\xi}{2\pi}\xi \bigg[ f\otimes f \, \vartheta^*{\,}^\psi\had^{(+)}_4 
\bigg]^\wedge(-\xi,\xi) \nnnl
&\,& \quad - \int_{\R^-}\frac{\rmd\xi}{2\pi} \xi \bigg[ f\otimes f
\,\vartheta^*\big(i\mathrm{S}_\spin-{\,}^\psi\had^{(+)}_4 \big) \bigg]^\wedge(-\xi,\xi) \, .
\end{eqnarray}
Here $\vartheta=\gamma\otimes\gamma$, $S_\spin$ is the fundamental solution to the Dirac equation, ${\,}^\psi\had^{(\pm)}_k$ are scalar distributions created from ${\,}^\psi H^{(\pm)}_k$ and $\hat{\,}$ denotes the Fourier transform which in our conventions is given by 
\begin{equation}
\hat{f}(\xi) = \int \rmd x \, f(x) \, e^{i\xi\cdot x} \, .
\end{equation}

The structure of this paper is as follows: In section \ref{dirac background} we present a review of the formulation of the classical (\S\ref{dirac background 1}-\ref{dirac background 2}) and quantum (\S\ref{quantum dirac hadamard}) Dirac fields and their Hadamard states. We then direct our attention to a microlocal description of the Hadamard series for the Dirac field in section \ref{microlocal review}; in particular we review the Sobolev wave-front set and its properties (\S\ref{microlocal review 1}) before applying the theory to the matter in hand and obtaining estimates on the singularities of the Hadamard series (\S\ref{microlocal results}). Finally, a point-splitting lemma is presented in section \ref{splitting section} before our main result is stated in section \ref{punch line}.

\section{The Dirac field in curved spacetime}\label{dirac background}

The reader who is familiar with the formalism necessary to describe the classical Dirac field on a curved background is encouraged to skip ahead to section \S\ref{quantum dirac hadamard}.

\subsection{Spin structures and spinors on curved spacetimes}\label{dirac background 1}

We begin by reviewing the geometry necessary to discuss the Dirac field
in a curved spacetime. We shall employ the algebraic framework for describing the Dirac quantum
field in a classical curved four-dimensional spacetime $(\mathcal{M},g)$.
Here $\mathcal{M}$ is a four-dimensional smooth manifold (assumed Hausdorff,
paracompact and without boundary) with a Lorentz metric $g_{ab}$ of signature
($+---$). Furthermore, we require $(\M,g)$ to be \textit{globally hyperbolic}, that is
$\M$ contains a Cauchy surface. Where index notation is
used, Latin indices will run over the range $0,1,2,3$ unless explicitly
stated otherwise, while Greek characters will denote frame indices and
also run over $0,1,2,3$ unless explicitly stated otherwise. We employ units
in which $c=\hbar=1$.

In Minkowski spacetime the spinors are nothing
more than the spin-half representation of the Poincar\'e group, however,
a general manifold does not exhibit this symmetry globally: Therefore, the usual (i.e. Minkowski spacetime) interpretation of a
spinor as being a quadruple of complex numbers at each point in spacetime does not
generalise under the replacement
$(\R^4,\eta)\mapsto (\M,g)$. A rigorous formulation of spinors on a manifold is
given in terms of fibre bundles where the spin group is the structure
group. (For a review of the necessary concepts related to fibre bundles, and in particular spinor bundles, the reader is directed to \cite{nakahara}). We shall review basic facts about the Dirac matrices and the
Lorentz and spin groups in Minkowski spacetime, and use a local frame to generalise the results to a curved spacetime. What follows is based on
\cite{dawson
fewster} and benefits from the
elaborations in \cite{dimock dirac, fewster verch}.

We begin by summarising several groups which appear in our discussion. The Lorentz group $O(1,3) = \{ \Lambda \in \mathrm{GL}_4(\R) \mid \eta_{\alpha\beta}\Lambda^\alpha_{\,\,\,\gamma}\Lambda^\beta_{\,\,\,\delta}=\eta_{\gamma\delta} 
\}$ has the subgroup $\mathfrak{L}^\uparrow_+$,
\begin{equation}
\mathfrak{L}^\uparrow_+ = \{ \Lambda \in O(1,3) \mid \det\Lambda=1 \,
\& \, \Lambda^0_{\,\,\,0}>0 \} 
\end{equation}
called the proper orthochronous Lorentz group. The Dirac gamma matrices $\gamma_\alpha$ satisfy the Clifford algebra relation  $\{\gamma_\alpha , \gamma_\beta  \} =
2\eta_{\alpha\beta}\mathbbm{1}$ and are said to belong to a \textit{standard
representation} if $\gamma^\dagger_0=\gamma_0$ and $\gamma^\dagger_{\alpha}=-\gamma_{\alpha}$ for $\alpha=1,2,3$.
From here on we shall assume that our Dirac matrices belong to a
standard representation.  The spin group, $\Spin(1,3)$, is defined by
\begin{equation}
\Spin (1,3) = \{ S \in \mathrm{SL}_4(\C) \mid S\gamma_\alpha S^{-1} =
\gamma_\beta \Lambda^\beta_{\,\,\,\alpha} \textrm{ for some }\Lambda\in\mathcal{L} 
\} \ ,
\end{equation}
and is known to be
a two-to-one cover of $\mathfrak{L}^\uparrow_+$, i.e. the mapping $S\mapsto \Lambda(S)$ is a two-to-one
covering homomorphism from the identity connected component
$\Spin_0(1,3)$, of $\Spin(1,3)$, to $\mathfrak{L}^\uparrow_+$ with kernel
$\{\mathbbm{1},-\mathbbm{1} \}$.

We now direct our attention to a curved spacetime setting. The frame bundle $\FM$ is the
bundle of oriented and time-oriented tetrads $\{ e_\alpha^a
\}_{\alpha=0,1,2,3}$ over spacetime $(\M,g)$ with the convention that
$e_0^a$ is a future pointing timelike vector; moreover, $\FM$ is a
principal $\mathcal{L}^\uparrow_+$ bundle whose right action is given by
$(R_\Lambda e)_\alpha = e_\beta \Lambda^\beta_{\,\,\,\alpha}$. A \textit{spin structure} on $(\M,g)$ is a principal
$\Spin_0(1,3)$ bundle, $\SM$, over $(\M,g)$ equipped with a fibre homomorphism
$\varphi:\SM\mapsto \FM$ such that $\varphi\circ
R_S=R_{\Lambda(S)}\circ\varphi$, i.e. $\varphi$ intertwines the right
action of the structure group on these bundles. Spin structures are not
unique, however two such structures, $\SM$ and $\widetilde{\SM}$ equipped
with $\varphi$ and $\widetilde\varphi$ respectively can be said to be equivalent if
there is an isomorphism $\iota : \SM \mapsto \widetilde{\SM}$ such that $\varphi=\widetilde{\varphi}\circ\iota$.

It is worth pointing out that spin structures
do not exist in general for an arbitrary manifold; their existence is
determined by the \textit{second Stiefel-Whitney class}. In essence, the requirement that the second Stiefel-Whitney class
vanishes ensures consistency between the (transition functions of the)
fibre group of the tangent bundle and the (lift to the transition functions of the) spin group. It is known that there exist spin
structures over orientable manifolds if and only if the second
Stiefel-Whitney class vanishes and that every four-dimensional globally hyperbolic manifold admits a spin
structure. We now assume that an arbitrary spin structure has been chosen and is
fixed for the remainder of this discussion.

We may now define spinor fields on a curved manifold by saying they are
sections of the associated $\Spin_0(1,3)$ bundle 
\begin{equation}
\DM = \SM \ltimes_{\Spin_0(1,3)}\C^4 \,  .
\end{equation}
The fibre of $\DM$ at $x\in\M$ is the equivalence class $[T,z]_x$ where $T\in S_x\M$
and $z\in\C^4$ is a \emph{column} vector and the equivalence relation is: $[\widetilde{T},\widetilde{z}]_x=[T,z]_x$ if and only
if $\widetilde{T}=R^{-1}_ST$ and $\widetilde{z}=Sz$ for some $S\in \Spin_0(1,3)$. The bundle
$D\M$ has fibre $\C^4$ at every point and left action given by $L_S
[T,z]_x = [T,Sz]_x$. The dual bundle
\begin{equation}
D^*\M=S\M\ltimes_{\Spin_0(\M,g)}{\C}^4 \, ,
\end{equation}
where ${\C}^4$
is the set of complex \emph{row} 4-vectors, is constructed similarly and its
fibres are the equivalence classes $[T,z^\T]_x^*$, $z^\T\in{\C}^4$ a row vector, such that $[\widetilde{T},\widetilde{z^\T}]_x^* =
[T,z^\T]_x^*$ if and only if $\widetilde{T}=R^{-1}_ST$ and $\widetilde{z^\T}=z^\T S^{-1}$ for some
$S\in\Spin_0(1,3)$. Just as the sections of $\DM$ are called \textit{spinors}, the sections of $D^*\M$ are called \textit{cospinors}. 
We shall refer to \textit{test spinors} as being the smooth and compactly supported sections of $D\M$, the space of which we denote $C^\infty_0(D\M)$; \textit{test cospinors} are similarly defined and are elements of $C^\infty_0(D^*\M)$. As expected, there
exists a natural pairing between spinors and cospinors: Set $v_x = [T,z_1^\T]_x^*$ and $u_x\in [T,z_2]_x$, then $v_x(u_x) = z_1^\T  z_2$ is a scalar.  

We are now in a position to define the Dirac adjoint operation ${\,}^+ :
D\M \mapsto D^*\M$ which is given by
\begin{equation}
[T,z]^+_x = [T,z^\dagger \gamma_0]^*_x \, .
\end{equation}

Any local section $E:\M\mapsto S\M$ of $S\M$ determines a local frame $e^a_\alpha$ by $\varphi\circ E$ and local sections $E_A$ of $D\M$, such that $E_A(x)=[E_x,z_A]$ where $\{ z_A \}_{A=0,1,2,3}$ is the canonical basis of
$\C^4$. The dual frames $e_a^\alpha$, $E^A$ are defined through $e^\alpha \cdot e_\beta=\delta^\alpha_\beta$ and
$E^A ( E_B)=\delta^A_B$. One may define a mixed tensor-spinor object $\gamma \in
C^\infty(T^*\M)\otimes C^\infty(D\M)\otimes C^\infty(D^*\M)$ by setting its components $\gamma^{\,\,\,A}_{\alpha\,\,\,\,\,B}$ in the frame
$e^\alpha_a\otimes E_A\otimes E^B$ equal to the matrix elements
$(\gamma_\alpha)^{A}_{\,\,\,\,B}$. For example, it can be shown that
\begin{equation}\label{split gamma 0}
\gamma_0=\delta^{AB}E_A\otimes E^+_B \, .
\end{equation}

\subsection{The Dirac equation}\label{dirac background 2}

The metric $g$ determines a connection $\Gamma$ in the usual way via the covariant derivative operator $\nabla: C^\infty(T\M)\mapsto C^\infty(T^*\M\otimes T\M)$. One may equally define a connection $\sigma$ and covariant derivative on the spinor and cospinor bundle, which we also denote by $\nabla$,
\begin{equation}
\nabla: \left\{  \begin{array}{c} C^\infty(D\M) \\ C^\infty(D^*\M)
\end{array}\right.  \mapsto \left\{  \begin{array}{c}
C^\infty(T^*\M\otimes D\M) \\ C^\infty(T^*\M\otimes D^*\M) \end{array}\right. \, .
\end{equation}
Given a local section $E$, $f\in C^\infty(D\M)$ may be decomposed
$f=f^AE_A$, then $\nabla_af$ has components
\begin{equation}
\nabla_\alpha f^A = \partial_\alpha f^A + \sigma^A_{\alpha B}f^B
\end{equation}
in the frame $e^\alpha_a\otimes E_A$ where the connection $\sigma$ has
elements given by $\sigma^A_{\alpha B} =
-\frac{1}{4}\Gamma^\beta_{\alpha\delta}\gamma_{\beta\,\,\,\,C}^{\,\,A}\gamma^{\delta
C}_{\,\,\,\,\,\, B}$.

We are now in a position to define the equation of motion the spinors
will satisfy, i.e. the Dirac equation. The Dirac operator 
\begin{equation}
\dirac : \left\{  \begin{array}{l} C^\infty(D\M) \\ C^\infty(D^*\M) \end{array}
\right. \mapsto
\left\{ \begin{array}{l} C^\infty (D\M) \\ C^\infty(D^*\M)
\end{array} \right.
\end{equation}
maps (co)spinor fields into (co)spinor fields by
\begin{eqnarray}
\dirac f &=& (\dirac f)^A E_A = \eta^{\alpha\beta}\gamma^{\,\,\,
A}_{\alpha\,\,\,\, B}(\nabla_\beta f^B) E_A \quad \forall f \in C^\infty(D\M)
\\
\dirac h &=& (\dirac h)_B E^B =
\eta^{\alpha\beta}(\nabla_\beta h_C)\gamma_{\alpha\,\,\,\, B}^{\,\,\, C} E^B
\quad \forall h \in C^\infty(D^*\M) \, .
\end{eqnarray}
The spinor field $f\in C^\infty(D\M)$ is said to satisfy the Dirac
equation if $(-i\dirac+\mu)f=0$ where the constant $\mu\geq 0$ is
interpreted as the mass of the field. Similarly, the cospinor field $h \in C^\infty(D^*\M)$ is said to
satisfy the Dirac equation if $(i\dirac+\mu)h=0$. 

Even though the Dirac operator is not normally hyperbolic it is possible
to find unique advanced and retarded fundamental solutions on arbitrary
globally hyperbolic spacetimes. The key element in this analysis is the
Lichn\'erowicz identity,
\begin{equation}\label{lich}
P = (-i\dirac +\mu)(i\dirac+\mu)
\end{equation}
where $P=\nabla^2+R/4+\mu^2$ is the
so-called supersymmetrically coupled Klein-Gordon operator for
spinors, which relates Dirac operators to normally hyperbolic ones. The $R$ featuring in (\ref{lich}) is the Ricci scalar. It
is known that there exist unique advanced $E^-_P$ and retarded
$E^+_P$ fundamental solutions to any normally hyperbolic operator $P$ on globally hyperbolic spacetimes.
Hence, for the spinor field, one has the following fundamental
solutions: $S^\pm_\textrm{sp}=(i\dirac+\mu)E^\pm_P$. To be explicit, $S^\pm_\textrm{sp}$ are continuous operators
$S^\pm_\mathrm{sp}: C^\infty_0(D\M)\mapsto C^\infty(D\M)$ such that 
\begin{equation}
(-i\dirac +\mu)S^\pm_\mathrm{sp} f =
S^\pm_\mathrm{sp}(-i\dirac+\mu)f =f 
\end{equation}
satisfying $\supp (S^\pm_\mathrm{sp} f)\subset J^\pm (\supp f)$.
Clearly, there is a similar construction for the cospinor field resulting in
$S^\pm_\textrm{cosp}$. The advanced-minus-retarded fundamental solution for
spinors is $S_\mathrm{sp}=S^-_\mathrm{sp}-S^+_\mathrm{sp}$ and for cospinors
$S_\mathrm{cosp}=S^-_\mathrm{cosp}-S^+_\mathrm{cosp}$. An additional
antilinear map $\Gamma :D^*\M\oplus D\M\mapsto D^*\M\oplus D\M$ acts by
\begin{equation}
\Gamma\left( \begin{array}{c} h \\ f \end{array} \right)
= \left( \begin{array}{c} f^+ \\ h^+\end{array} \right) \,  ,
\end{equation}
where  $f\in C^\infty_0(D\M)$ and $h\in C^\infty_0(D^*\M)$. The map $\Gamma$ makes explicit the symmetry between spinors and
cospinors in this dual setting.

\subsection{The quantum Dirac field and Hadamard states}\label{quantum dirac hadamard}

To define the field algebra we denote by $\mathcal{D}(D^*\M \oplus D\M) =
C^\infty_0(D^*\M)\oplus C^\infty_0(D\M)$ the space of all test cospinors and
test spinors on which the operators 
\begin{equation}
D := i\left(   \begin{array}{cc} i\dirac +\mu & 0 \\ 0 & -i\dirac + \mu \end{array}\right) \, \quad S := i\left(   \begin{array}{cc}
S_\mathrm{cosp} & 0 \\ 0 & S_\mathrm{sp}  \end{array}\right)
\end{equation}
act. The elements $\mathcal{F}\in\mathcal{D}(D^*\M\oplus D\M)$ may be used to label a set of abstract
objects $\{ \Psi(\mathcal{F}) \mid \mathcal{F} \in
\mathcal{D}(D^*\M\oplus D\M)\}$ which, when equipped with $\mathbbm{1}$, generates a unital *-algebra $\mathfrak{F}$. We define
the \textit{algebra of smeared fields} $\mathfrak{F}(\M,g)$ to be
$\mathfrak{F}$ quotiented by \\
\begin{tabular}{rl}
i) & Adjoint, $\Psi (\mathcal{F})^* = \Psi(\Gamma\mathcal{F})$; \\
ii) & Linearity, $\Psi(\alpha_1 \mathcal{F}_1 + \alpha_2\mathcal{F}_2) =
\alpha_1\Psi(\mathcal{F}_1)+\alpha_2\Psi(\mathcal{F}_2)$; \\
iii) & The field equation, $\Psi(D\mathcal{F})=0$; \\
iv) & Canonical anticommutation relation, $\{\Psi(\mathcal{F}_1),\Psi(\mathcal{F}_2) \}=iS(\mathcal{F}_1,\mathcal{F}_2)\mathbbm{1}$. 
\end{tabular} \\
Here $\mathcal{F},
\mathcal{F}_1,\mathcal{F}_2\in\mathcal{D}(D^*\M\oplus D\M)$ and
$\alpha_1,\alpha_2\in\C$. It is relation (iv) that quantises the theory.

The usual Dirac field $\psi$ and its adjoint $\psi^+$ are special cases of the above
construction. For $h\in C^\infty_0(D^*\M)$, $f\in C^\infty_0(D\M)$ we define 
\begin{equation}
\psi(h) := \Psi\left( \begin{array}{c} h \\ 0 \end{array}  \right)  \quad
\textrm{and} \quad \psi^+(f) = \Psi \left( \begin{array}{c} 0 \\ f \end{array}  \right)
\end{equation}
which we interpret as smeared fields.

A \textit{state} $\omega : \mathfrak{F}(\M,g)\mapsto \C$ is a linear functional
which is positive, i.e. $\omega(A^*A)\geq 0$ $\forall A\in\mathfrak{F}(\M,g)$, and
normalised such that $\omega(\mathbbm{1})=1$. We shall restrict our attention to states for which the two-point function, defined by 
$$\omega(\Psi(\mathcal{F}_1)\Psi(\mathcal{F}_2)) \quad \forall \mathcal{F}_1,\mathcal{F}_2 \in \D(D^*\M\oplus D\M) \, ,$$ is a distribution on
$\mathcal{D}(D^*\M \oplus D\M)\otimes \mathcal{D}(D^*\M \oplus D\M)$. Associated to each state $\omega$ we define what we shall call the Dirac and Dirac adjoint two-point functions
$\cW_\omega, \cW^\Gamma_\omega \in
\D'(D\M\times D^*\M)$ respectively by 
\begin{equation}
\cW_\omega(f,h)=\omega
(\psi^+(f) \psi(h)) \quad \textrm{and} \quad
\cW^\Gamma_\omega(f,h)=\omega(\psi(h)\psi^+(f)) \, .
\end{equation}
As a consquence of the positivity of states we immediately have that 
\begin{equation}
\cW(f,f^+) \geq 0 \quad \forall f \in C^\infty_0(D\M) \quad \textrm{and} \quad \cW^\Gamma(h^+,h)\geq 0 \quad \forall h\in C^\infty(D^*\M) \, .
\end{equation}
The covariant anticommutation relation in terms of $\cW_\omega$ and $\cW^\Gamma_\omega$ is equally expressed
\begin{equation}
\cW_\omega + \cW^\Gamma_\omega = iS_\spin \, .
\end{equation}

It is clear from our discussion concerning the advanced and retarded
fundamental solutions of the Dirac field that one may use the
Lichn\'erowicz identity to define a notion of Hadamard state for the
Dirac field. In order to give the precise formulation of the Hadamard series
construction we must first discuss some geometry and here we follow
\cite{sahlmann2}. We denote by $\mathfrak{X}\subset \M\times\M$ the
set
\begin{eqnarray}
\mathfrak{X} &=& \{ (x,x^\prime)\in\M\times\M \mid  x,x^\prime \text{
are causally
related and} \nnnl &\,& \qquad J^+(x)\cap J^-(x^\prime) \,\textrm{and}\, J^-(x)\cap J^+(x^\prime)
\text{ are
contained} \nnnl &\,& \qquad \text{within a convex normal neighbourhood} 
\} \, .
\end{eqnarray}
Let $X\subset \mathfrak{X}$ be an open subset of
$\mathfrak{X}$ such that between any pair $(x,x^\prime)\in X$ there
exists a unique geodesic connecting them such that the
(signed) geodesic
separation of points defines a smooth function $\sigma$ on $X$. We make the additional
requirement that the Hadamard construction (to be described shortly) can be carried out on $X$. Subject to all these requirements, we call $X$ \label{notation reg dom}a \textit{regular
domain}. We define
two sequences of distributions $\{ H_k^{(\pm)}\}_{k=0,1,2,\dots} \in
\D^\prime (X)$ by 
\begin{eqnarray}\label{parametrix}
H_k^{(\pm)}(x,x^\prime) &=&  \frac{1}{4\pi^2} \bigg\{ \frac{\Delta^\frac{1}{2}(x,x^\prime)}{\sigma_\pm (x,x^\prime)} + \sum^k_{j=0}v_j
(x,x^\prime)\frac{\sigma^j(x,x^\prime)}{\ell ^{2(j+1)}}\ln\bigg(\frac{\sigma_\pm (x,x^\prime)}{\ell^2}\bigg)
\nonumber \\
&\,& \qquad 
+\sum^k_{j=0}w_j(x,x^\prime)\frac{\sigma^j(x,x^\prime)}{\ell^{2(j+1)}}\bigg\}
\, 
\end{eqnarray}
where we have introduced a length scale $\ell$ to make
$\sigma/\ell^2$ dimensionless.  By $F(\sigma_\pm)$, $F$ some function, we mean 
\begin{equation}
F(\sigma_\pm) = \lim_{\epsilon\rightarrow 0^+}F(\sigma_{\pm\epsilon}) 
\end{equation}
in the sense of distributions, where $\sigma_{\pm\epsilon}(x,x^\prime) =
\sigma(x,x^\prime)\pm 2i\epsilon(t(x)-t(x^\prime))+\epsilon^2$ and $t$ is a
time function on $X$. The functions $\Delta$, known as
the van Vleck-Morette determinant, $v_j$ and $w_j$ are found by fixing $x^\prime$
and applying $P\otimes\mathbbm{1}$ to $H_k^{(+)}$ and equating all the
coefficients of $1/\sigma_+$, $1/\sigma^2_+$, $\ln \sigma_+$ etc to
zero; moreover, they are all spinors (i.e., they carry internal indices). This
determines a system of differential equations known as the \textit{Hadamard
recursion relations}\footnote{The Hadamard recursions relations for the
scalar field can be found in \cite{fewster smith}.}. In $X$ the system of differential equations uniquely determines the $\{ v_j \}_{j=0,\dots,k}$
series. The $\{ w_j\}_{j=0,\dots,k}$ series is specified once the value of $w_0$ is
fixed; we adopt Wald's prescription that $w_0=0$ \cite{wald}.

Let $u$ be of Hadamard form for the operator $P$, i.e.
within a regular domain $X$ one has $u = H_k^{(+)}$ modulo
$C^k(X)$ for each $k\in\N$; the distribution $u$ is sometimes referred to
as the \textit{auxiliary two-point function}. A state $\omega$ on the algebra of smeared (Dirac) fields is
said to be \textit{Hadamard} if its associated two-point function $\cW_\omega$ is of the form
$\cW_\omega = (i\dirac+\mu)\otimes\mathbbm{1}u$ where $u$ is an
auxiliary two-point function. Consequently, we have, within a regular domain
$X$, that
\begin{equation}
\cW_\omega = {\,}^\psi H_k^{(+)} \quad \textrm{modulo}  \quad C^k(X)
\end{equation}
where ${\,}^\psi H^{(+)}_k =
(i\dirac+\mu)\otimes\mathbbm{1}H^{(+)}_{k+1}$. We also define ${\,}^\psi H^{(-)}_k = (i\dirac+\mu) \otimes\mathbbm{1}H^{(-)}_{k+1}$. Hence, within a regular domain
$X$, for any Hadamard state $\omega$ we have the following identities:  
\begin{eqnarray}\label{dirac hadamard relations}
\cW_\omega &=& {\,}^\psi H_k^{(+)} \quad \textrm{modulo } C^k(X) \, , \label{dirac hadamard relations 1} \\
\cW^\Gamma_\omega &=& -{\,}^\psi H^{(-)}_k \quad \textrm{modulo } C^k(X) \label{dirac hadamard relations 2}\, ,
\label{sign note}\\
iS_\spin &=& {\,}^\psi H^{(+)}_k - {\,}^\psi H^{(-)}_k \quad \textrm{modulo } C^k(X)  \, . \label{dirac hadamard relations 3}
\end{eqnarray}
Two remarks are in order: First, note that we require the sign in (\ref{sign note}) so as to ensure
that the anticommutation relation holds. Second, observe that as $k$ increases $\cW_\omega - {\,}^\psi H_k^{(+)}$ becomes more regular and that, for sufficiently high $k$, $\cW_\omega - {\,}^\psi H_k^{(+)}$ has a well defined coincidence limit.

\subsection{The stress energy tensor}\label{energy section}

We open this section with a few remarks about obtaining scalar distributions from spinorial ones as this will be the
basis of our analysis of the energy density for the remainder of our
discussion. Let $E_A$ be the spinor field derived from a local section $E$ of the
spin bundle $\SM$. Then one can derive matrices $\cW_{\omega AB}$ and $\cW^\Gamma_{\omega AB}$ from $\cW_\omega$ and $\cW^\Gamma_\omega$ via
\begin{equation}
\cW_{\omega AB}(f,f') = \cW_\omega(fE_A,f'E_B^+) \quad \cW_{\omega AB}^\Gamma (f,f') = \cW^\Gamma_\omega (fE_A,f'E^+_B)
\end{equation}
for all $f,f'\in C^\infty_0(\M)$. Scalar bi-distributions $\W_\omega$
and $\W^\Gamma_\omega$ (in $\D'(\M\times\M)$) may be constructed by taking the traces of the matrices, i.e.
$\W_\omega=\delta^{AB}\cW_{\omega AB}$ and $\W^\Gamma_\omega
= \delta^{AB}\cW^\Gamma_{\omega AB}$. Similarly, one may define a scalar
version ${\,}^\psi\had_k^{(\pm)}$ of the Hadamard series ${\,}^\psi H_k^{(\pm)}$. 

The stress energy tensor of the classical spin-1/2 field $\psi$ is
given by 
\begin{equation}
T_{ab} = \frac{i}{2}\bigg( \psi^+\gamma_{(a}\nabla_{b)}\psi - (\nabla_{(a}\psi^+)\gamma_{b)}\psi 
\bigg) \, 
\end{equation}
where the subscript parentheses denote symmetrisation, i.e $\tau_{(ab)}=(\tau_{ab}+\tau_{ba})/2$. As advertised, we shall
concentrate exclusively on proving an absolute QWEI along a timelike worldline. 
Therefore, we pick a properly parametrised smooth timelike worldline $\gamma:\R\mapsto \M$ and consider a
spacetime tube $\tau_\gamma\subset \M$ centred about it, a precise construction of this tube will be given shortly. In
$\tau_\gamma$ we may construct a tetrad $\{e^a_\alpha
\}_{\alpha=0,1,2,3}$ such that $e^a_0 = \dot{\gamma}^a$ and the remaining $\{ e^a_\alpha\}_{\alpha=1,2,3}$ are orthonormal
to this vector. Then, in the (dual) frame $e^\alpha_a\otimes e^\beta_b$, the tensor $T_{ab}$ has components $T_{\alpha\beta}$ given by
\begin{equation}
T_{\alpha\beta} = \frac{i}{2}\bigg(
\psi^+\gamma_{(\alpha}\nabla_{\beta)}\psi - (\nabla_{(\alpha}\psi^+)\gamma_{\beta)}\psi  
\bigg) \, 
\end{equation}
from which follows the energy density $\rho$:
\begin{equation}
\rho = \frac{i}{2}\bigg(
\psi^+\gamma_0\nabla_{0}\psi - (\nabla_{0}\psi^+)\gamma_{0}\psi  
\bigg) \, .
\end{equation}

Recall from (\ref{split gamma 0}), that the spin frame $E_A$ and its Dirac conjugate satisfy $\delta^{AB}E_A\otimes E_B^+ =\gamma_0$. This enables
us to write the classical point-split stress energy density $\rho^\mathrm{split}$ of the Dirac field as
\begin{eqnarray}
\rho^\mathrm{split}(x,x') &=& \frac{i}{2}\delta^{AB}\bigg( \big( \psi^+E_A  \big)\otimes \big( E^+_B e_{0'}\cdot\nabla\psi\big) \nnnl
&\,& \qquad - \big( [e_0\cdot\nabla\psi^+]E_A \big)\otimes \big( E^+_B\psi \big)
\bigg)(x,x') \, .
\end{eqnarray}
One may then use $\rho^\mathrm{split}$ to define the distributional
point-split energy density, also denoted $\rho^\mathrm{split}$: Let
$f,f' \in C^\infty_0(\M)$ then 
\begin{eqnarray}
\rho^\mathrm{split}(f',f) &=& \frac{i}{2}\delta^{AB}\int_{\M\times\M}\dvol(x)\dvol(x') \nnnl
&\,& \quad (f'\otimes f)\bigg( \big( \psi^+E_A  \big)\otimes \big( E^+_B e_{0'}\cdot\nabla\psi\big) \nnnl
&\,& \qquad - \big( [e_0\cdot\nabla\psi^+]E_A \big)\otimes \big( E^+_B\psi \big)
\bigg)(x,x') \\
&=& \frac{i}{2}\delta^{AB} \bigg(
\psi^+(f'E_A)\psi(\nabla\cdot(e_{0'}fE^+_B)) \nnnl 
&\,& \qquad - \psi^+(\nabla\cdot(e_0f'E_A))\psi(fE^+_B) \bigg)
\label{class rho split dirac}
\end{eqnarray}
where we have expressed the right hand side in more traditional
distributional language and $u\nabla\cdot v$ denotes the distributional dual of $(\nabla
u)\cdot v$. It is now clear that the replacement of $\psi^+\otimes\psi$ in
(\ref{class rho split dirac}) by the two-point function of any
Hadamard state $\omega$ defines the expectation point-split energy
density $\langle \rho^\mathrm{split}\rangle_\omega$ of the
Dirac field in that state. For notational convenience, we
decompose $\langle \rho^\mathrm{split}\rangle_\omega$ into the operator
$(T^\mathrm{split}_{00'})^{AB}$ acting on $\cW_{\omega AB} - {\,}^\psi
H_{1 AB}^{(+)}$ where $(T^\mathrm{split}_{00'})^{AB}$ is given
by
\begin{eqnarray}
(T^\mathrm{split}_{00'})^{AB} = \frac{i}{2}\delta^{AB}\bigg(\mathbbm{1} \otimes
e_0\cdot\nabla - e_0\cdot\nabla\otimes\mathbbm{1}   \bigg)+ \Theta^{AB}
\end{eqnarray}
and $\Theta^{AB}$ is a term which depends only on the spin-connection.
The precise form of $\Theta^{AB}$ may be found in \cite{fewster verch} (eqn. (3.10) of
that reference) but does not affect our discussion due to a useful result\footnote{Lemma 4
of \cite{dawson fewster}.} which shows it is identically zero on a timelike
worldline under certain conditions which we
shall now motivate: Let
$\gamma:\R\mapsto \M$ be a timelike worldline with unit tangent vector
$\dot\gamma$. Pick a point $x$ on $\gamma\subset\M$ and construct a
local frame $\{ e^a_\alpha  \}_{\alpha=0,1,2,3}$ subject to
$e^a_0=\dot\gamma^a$ at $x$. As we shall only be concerned with averaging
along a compact subset of $\gamma$ we shall fix a closed interval
$I\subset\gamma(\R)$ such that $x\in I$. One may utilise Fermi-Walker transport to move
$\{ e^a_\alpha \}_{\alpha=0,1,2,3}$ along $I$ keeping
$e^a_0|_\gamma=\dot\gamma^a$. The salient feature of Fermi-Walker
transport is that it preserves angles, i.e. $\{
e^a_\alpha\}_{\alpha=0,1,2,3}$ remains an orthonormal family along $I$. Next, at each point $\gamma(s)$ along
$I$ consider the convex normal neighbourhood $\U$ orthogonal to $\dot\gamma(s)$
and for each $y\in\U$ parallel transport $\{ e^a_\alpha \}_{\alpha=0,1,2,3}$ along the
unique geodesic connecting $y$ to $\gamma(s)$. In this manner we `sweep out' a tube
$\tau_\gamma\subset \M$ in spacetime. Importantly, the local frame
(throughout $\tau_\gamma)$ is a local section of $F\M$ which may be
identified with a smooth section $E$ of $S\M$. The details of this
identification may be found in \S3 of \cite{fewster verch}. For our
purposes it is sufficient to know that, as a consequence of this
construction, the form of $(T^\mathrm{split}_{00'})^{AB}$ simplifies
when restricted to the diagonal. The precise statement, quoted from
\cite{dawson fewster} (lem.4) is:

\begin{lemma}
If $E$ is any local section of $S\M$ obtained in the above fashion
from the curve $\gamma$, then $\Theta^{AB}|_\gamma=0$.
\end{lemma}

We shall assume that $E$ has been obtained in this way. Therefore, the
finite contribution $\langle \rho^\fin\rangle_\omega$ to the energy
density is given by:
\begin{equation}
\langle \rho^\fin  \rangle_\omega (\gamma(t)) = \frac{i}{2}\vartheta^*\bigg(\big( \mathbbm{1}\otimes e_0\cdot\nabla -e_0\cdot\nabla \otimes \mathbbm{1} \big)
\big(\W_\omega - {\,}^\psi\had^{(+)}_1  \big)\bigg)(t,t) 
\end{equation}
where $\vartheta=\gamma\otimes\gamma$. Finally, it may be argued in
analogy with \cite{fewster verch} that one may re-express $\langle \rho^\fin\rangle_\omega$ as
\begin{equation}\label{Had regularised energy}
\langle \rho^\fin\rangle_\omega (\gamma(t)) = \frac{1}{2}\bigg( \mathbbm{1} \otimes D - D\otimes\mathbbm{1}  \bigg)\vartheta^*\big( \W_\omega-{\,}^\psi\had^{(+)}_1 \big)(t,t)
\end{equation}
where $D$ is the distributional dual to $-i\rmd/\rmd t$.

\section{Microlocal analysis applied to quantum field theory}\label{microlocal review}

\subsection{The Sobolev wave-front set}\label{microlocal review 1}

Since the publication of Radzikowski's equivalence theorem
\cite{radzikowski} microlocal analysis, in particular H\"ormander's
concept of wave-front set, has been successfully applied to quantum
field theory\footnote{For a readable account of the general significance of
microlocal analysis in quantum field theory the reader is directed to \cite{wald review}.}. The proof of the most general QWEIs rely on microlocal
analysis at various stages in their argument, e.g. \cite{dawson fewster, fewster, fewster
pfenning}. A refined version of the usual (smooth) wave-front set, an
exposition of which may be found in chapter VIII of \cite{hormander},
has already been employed in the proof of an absolute quantum energy
inequality for the Klein-Gordon field \cite{fewster smith}. We shall
briefly review the necessary details of this refinement, known as the
Sobolev wave-front set.

For $s \in \R$, the \textit{Sobolev space} $H^s(\R^n)$ is the set of all tempered
distributions $u$ such that $\widehat{u}$ is a measurable function and 
\begin{equation}
\norm u \norm_{H^s(\R^n)}^2=\int_{\R^n} \rmd^n\xi \, (1+|\xi|^2)^s |\widehat{u}(\xi)|^2 < \infty \, .
\end{equation}
It is clear that
$\s(\R^n)\subset H^s(\R^n)$ for each $s\in\R$. Moreover, one may
show that $\s(\R^n)$ is dense in
$H^s(\R^n)$, see chapter 1 \S3 of \cite{taylor} for a brief argument.

Sobolev space theory is usually introduced into the study of
distributional solutions to partial differential equations by asking
when such a solution is an honest function; this is the subject of the embedding
theorems. We summarise the following useful properties of
the $H^s$ spaces, the first of which is a relevant embedding theorem:

\begin{prop}\label{Sobolev prop}
The Sobolev spaces $H^s(\R^n)$ have the following properties:
\begin{equation*}
\begin{array}{cl}
i) & \textrm{Let $k\in \{ 0\} \cup\N$ and $s\in\R$ satisfy $s>k+n/2$ then
} H^s(\R^n) \subset C^k(\R^n) \\ & \text{
is a continuous embedding} \\
ii) & H^s(\R^n) \subset H^{s^\prime}(\R^n) \, \forall s\geq s^\prime; \\
iii) & \textrm{if } u \in H^s(\R^n), f \in C^k(\R^n) \textit{ and } D^\alpha f  \in
L^\infty(\R^n) \, \forall |\alpha|\leq k, \textit{ where $D$ is a} \\
& \textit{partial derivative operator and $\alpha$ is a multi-index, then for all } |s|\leq
k \\ &  
u \mapsto fu \textit{ is a bounded linear map of $H^s(\R^n)$ into $H^s(\R^n)$. In particular,} \\
& \textit{$H^s(\R^n)$ is closed under multiplication by smooth functions.}
\end{array}
\end{equation*}
\end{prop}

The concept of the Sobolev wave-front set will give a concise way
of saying what it means for a distribution to microlocally fail to be an element of a
Sobolev space. For convenience, we adopt the notation that $\dot{T}^*\R^n$ (similarly $\dot{T}\R^n$, $\dot{T}\M$ etc) is the bundle $T^*\R^n$ ($T\R^n$, $T\M$, etc) with the zero section removed.

\begin{defn}
A distribution $u \in \D^\prime (\R^n)$ is said to be microlocally $H^s$ at
$(x,\xi)\in \dot{T}^*\R^n$ if there exists an open cone $\Gamma \subset \R^n\setminus 0$ about $\xi$ and a smooth function $\varphi \in C^\infty_0 (\R^n)$, $\varphi(x)\not=0$, such that 
\begin{equation}
\int_\Gamma \rmd^n\zeta \, (1+|\zeta|^2)^s |[\varphi u]^\wedge(\zeta)|^2 < \infty \, .
\end{equation}
The Sobolev wave-front set $WF^s(u)$ of a distribution $u \in \D^\prime (\R^n)$ is the complement, in $\dot{T}^*\R^n$, of the set of all pairs $(x,\xi)$ at which $u$ is microlocally $H^s$.
\end{defn}

To define the Sobolev wave-front set of a distribution on a
manifold one works locally. Let $\U$ be an open patch of a manifold $\M$
with associated coordinate map $\kappa:\U\mapsto \R^n$. If
$(\kappa(x),\xi)\in WF^s(u\circ\kappa^{-1})$ then
$(x,\kappa^{-1}_*(\xi))\in WF^s(u)\subset \dot{T}^*\M$. We shall occasionally use the notation $u\in H^s_\loc(\M)$ if  $WF^s(u)=\emptyset$ for a distribution $u\in\D'(\M)$ and direct the reader to the remarks following definition 8.2.5 of \cite{hormander hyper} to justify this notation.

The Sobolev wave-front set is a closed cone in $\dot{T}^*\M$. Furthermore, we have the following properties\footnote{Taken
from the remarks following definition B.1 of \cite{junker}.} of $WF^s$: \\
\begin{tabular}{rl}
i) & The smooth wave-front set is related to the Sobolev wave-front set
\\
& \quad via $WF(u)=\overline{\bigcup_{s\in\R} WF^s(u)}$. \\
ii) & If $\varphi \in C^\infty_0 (\R^n)$ does not vanish in a
neighbourhood of $x$ then \\
& \quad $(x,\xi)\in WF^s(u)$ if and only if
$(x,\xi)\in WF^s(\varphi u)$. \\
iii) & $(x,\xi)\in WF^s(u)$ if and only
if, for all $v\in H^s_\loc$, $(x,\xi)\in WF^s(u-v)$. \\
iv) & $WF^s(u+w)\subset WF^s(u)\cup
WF^s(w)$. \\
v) & The nesting property: $WF^s(u)\subset WF^{s'}(u)$ $\forall
s\leq s'$.
\end{tabular}

It is also possible to see explicitly in $WF^s$ what effect partial differential
operators have on the singularities of distributions. For a general $m-$dimensional smooth manifold $\M$, let $P$ be a partial differential operator of
order $r$, i.e. in local coordinates on $\M$
\begin{equation}
P = \sum_{|\alpha|\leq r} p_\alpha(x)(-i\partial_a)^{\alpha} 
\end{equation} 
where $\alpha$ is a multi-index and $p_\alpha$ are smooth functions, then the \textit{principal symbol}, $p_r(x,\xi)$, of
$P$ is 
\begin{equation}
p_r(x,\xi) = \sum_{|\alpha|=r}p_\alpha(x)\xi^{\alpha}_a
\, .
\end{equation}
The \textit{characteristic set}, $\Char P$, of a partial differential operator $P$ is the set of
$(x,\xi)\in \dot{T}^*\M$ such that the principal symbol
vanishes. We may now quote corollaries 8.4.9-10 of \cite{hormander hyper} which show the effect
differential operators have on the Sobolev wave-front set of a
distribution:

\begin{lemma}\label{order}
Let $\M$ be a smooth manifold. For $u \in \D^\prime (\M)$ and any
linear partial differential operator $P$
of order $r$ with smooth coefficients then $WF^s(Pu)\subset WF^{s+r}(u)$ and $WF^{s+r}(u)\subset WF^{s}(Pu)\cup \Char P$.
\end{lemma}

We close this section with the statement of Beal's restriction theorem, which tells us under what circumstances a distribution may be restricted to a submanifold, and a result about the implications for the positivity of states under such a restriction. Such results are of interest to us as we need to understand how to restrict those distributions which make up the Hadamard series (and ones derived from it such as the
point-split energy density) to timelike worldlines. 

Beal's restriction theorem tells us that, for certain well behaved restrictions, the Sobolev grading on
the wave-front set is reduced by an amount proportional to the
codimension of the restriction. The result discusses the case of restricting a distribution on a $m$ dimensional manifold $\M$ to a
smoothly embedded submanifold
$\Sigma$ of dimension $n$, writing the embedding as $\iota : \Sigma \rightarrow \M$. The embedding function $\iota$ has associated conormal bundle $N^*\Sigma$ given by
\begin{equation}
N^* \Sigma = \{ (\iota(x),\xi)\in T^*\M ; \, x \in \Sigma, \, \iota^*(\xi) = 0 \} \,
.
\end{equation}

We wish to formulate a statement of the restriction theorem for product
manifolds\footnote{The following result is adapted from lemma 11.6.1 of 
\cite{hormander hyper} which is a refinement of the standard restriction theorem
which may be found presented
as theorem 8.2.4 of \cite{hormander}. We have used the notation of prop.
B7 of \cite{junker}.}. As usual we let $(\M,g)$ denote a smooth $m$
dimensional spacetime, $\Sigma\subset \M$ an $n\leq m$ dimensional
submanifold embedded using $\iota : \Sigma \mapsto \M$. Then we define the
map $\vartheta : \Sigma \times\Sigma \mapsto \M \times \M$
by $\vartheta = \iota\otimes\iota$, the pull back $\vartheta^*$ may
sometimes be referred to as a \textit{restriction map}.

\begin{theorem}[Beal's Restriction theorem]\label{beal}
Let $u \in \D^\prime (\M\times\M)$ and
$\vartheta$ be defined as above. If
$\big( N^*\Sigma\times N^*\Sigma\big)\cap WF^s(u)=\emptyset$ for some $s>m-n$ then the restriction $\vartheta^*u$ of $u$ to
$\Sigma\times\Sigma$ is a well defined distribution in
$\D^\prime(\Sigma\times\Sigma)$. Moreover, 
\begin{equation}
WF^{s-(m-n)}(\vartheta^* u)\subset \vartheta^*WF^s(u)
\end{equation}
where the set $\vartheta^*WF^s(u)$ is defined to be
\begin{eqnarray}
\vartheta^*WF^s (u) &=& \{ (t,\iota^*(\xi);t^\prime,\iota^*(\xi^\prime)) \in
(T^*\Sigma\times T^*\Sigma) \mid \nonumber \\ 
&\,& \qquad (\iota(t),\xi;\iota(t^\prime),\xi^\prime) \in WF^s(u)  \}
\, .
\end{eqnarray}
\end{theorem}

Finally, we state a result\footnote{Theorem 2.2 of \cite{fewster}.} which asserts that the
positivity of states is preserved under the restrictions carried out
by Beal's theorem.

\begin{lemma}\label{res pos}
Let $\M$ and $\Sigma$ be smooth manifolds each equipped with smooth positive
densities, and suppose
$\iota:\Sigma\mapsto\M$ is smooth. If $u\in\D'(\M\times\M)$ is positive in the sense of states and $WF(u)\cap (N^*\Sigma\times N^*\Sigma)=\emptyset$,
then $\vartheta^*u=u\circ(\iota\otimes\iota)\in\D'(\Sigma\times\Sigma)$ is also positive.
\end{lemma}

\subsection{A microlocal description of the Hadamard series}\label{microlocal results}

Denote by $\mathcal{R} = \{ (x,\xi)\in \dot{T}^*\M \, \mid \,
g^{ab}(x)\xi_a\xi_b =0  \}$ the bundle of null covectors over $\M$. Since $(\M,g)$ is time orientable
we may decompose $\mathcal{R}$ into two disjoint sets $\mathcal{R}^\pm$
defined by $\mathcal{R}^\pm = \{ (x,\xi) \in \mathcal{R} \, \mid \,
\pm\xi \rhd 0 \}$ where by $\xi\rhd 0$ ($\xi\in T^*_x\M$) we mean that $\xi_a$ is in the
dual of the future light cone at $x$. We define the notation
$(x,\xi)\sim (x',\xi')$ to mean that there
exists a null curve $\gamma : [0,1]\mapsto \M$ such that $\gamma(0)=x$,
$\gamma(1)=x^\prime$ and $\xi_a = \dot{\gamma}^b(0)g_{ab}(x)$,
$\xi^\prime_a = \dot{\gamma}^b(1)g_{ab}(x^\prime)$. In the instance
where $x=x'$, $(x,\xi)\sim (x,\xi')$ shall mean that $\xi=\xi'$ is null. Then,
for convenience, define the set 
\begin{equation}
C = \{ (x,\xi;x',\xi') \in \mathcal{R} \times \mathcal{R}  \mid (x,\xi) \sim (x',\xi')  
\} \, .
\end{equation}
The set $C^{+-}$ is defined to be 
\begin{equation}\label{C+-}
C^{+-} = \{ (x,\xi;x^\prime,-\xi^\prime) \in
C \mid  \xi
\rhd 0 \} \, . 
\end{equation}
An occasionally useful set will be $C^{-+}$ defined by 
\begin{equation}
C^{-+} = \{ (x,-\xi;x^\prime,\xi^\prime) \in
C \mid  \xi\rhd 0 \} \, .
\end{equation}

Junker \& Schrohe \cite{junker} have proven that the quantum Klein-Gordon field, whose two-point function we denote $\Lambda_\omega$, obeys the following condition for all Hadamard states 
\begin{equation}
WF^s(\Lambda_\omega) = \left\{ \begin{array}{cl} C^{+-} & s\geq -1/2 \\ \emptyset & s<-1/2  \end{array}  \right. \, .
\end{equation}
In \cite{fewster smith} an analysis of the Hadamard series $H^{(+)}_k$ was given and concluded that 
\begin{equation}
WF^{s+j+1}(\sigma^j\ln\sigma_+)\subset
WF^s(1/\sigma_+) = \left\{ \begin{array}{cc} C^{+-} & s \geq -1/2 \\
\emptyset & s<-1/2    \end{array}  \right. \, 
\end{equation}
for $j\in\{0\}\cup\N$ and where $1/\sigma_+$ and $\sigma^j\ln\sigma_+$ are the singular constituents of the Hadamard series $H^{(+)}_k$ (\ref{parametrix}). Since it is known that the the coefficients $\Delta^\frac{1}{2}$ and $v_j$ appearing in the $H^{(+)}_k$ series are symmetric it follows from the simple symmetry argument $H^{(+)}_k(x,x')=H^{(-)}_k(x',x)$ modulo smooth functions that 
\begin{equation}
WF^{s+j+1}(\sigma^j\ln\sigma_\pm)\subset
WF^s(1/\sigma_\pm) = \left\{ \begin{array}{cc} C^{\pm\mp} & s \geq -1/2 \\ 
\emptyset & s<-1/2    \end{array}  \right. \, .
\end{equation}

Therefore, lemma \ref{order} implies that if $\omega$ is a Hadamard state for the Dirac field then within a regular domain
\begin{equation}
WF^s(\cW_\omega) \subset \left\{  \begin{array}{cl} C^{+-} &  s\geq 1/2 \\ \emptyset & s<1/2 \end{array}\right. \, .
\end{equation}
A similar condition holds for $\cW^\Gamma_\omega$ under the replacement $\cW_\omega \mapsto \cW^\Gamma_\omega$ and $C^{+-}\mapsto C^{-+}$.  Moreover, one may use the relations (\ref{dirac hadamard relations 1}-\ref{dirac hadamard relations 3}) to conclude that within a regular domain
\begin{eqnarray}
WF^s(\cW_\omega - {\,}^\psi H^{(+)}_k) &\subset& \left\{
\begin{array}{cc} C^{+-} & s \geq k+3/2 \\ \emptyset & s < k+3/2   \end{array}  \right. \, ,\\
WF^s(\cW_\omega^\Gamma - {\,}^\psi H^{(-)}_k) &\subset& \left\{
\begin{array}{cc} C^{-+} & s \geq k+3/2 \\ \emptyset & s < k+3/2   \end{array}  \right.
\, ,
\end{eqnarray}
where $\omega$ is a Hadamard state. Consequently, if $\omega$ is a Hadamard
state then (because $\W_\omega,\W^\Gamma_\omega$ are formed
from linear combinations of $\cW_\omega$ and $\cW^\Gamma_\omega$) it follows that the Sobolev wave-front set conditions which apply to $\cW_\omega$ and $\cW^\Gamma_\omega$ also apply to
$\W_\omega$ and $\W^\Gamma_\omega$ respectively. We encapsulate these findings in the following corollary:

\begin{corollary}\label{dirac grading}
Let $\omega$ be a Hadamard state for the Dirac field; then within a regular domain
\begin{eqnarray}
WF^s(\W_\omega - {\,}^\psi \had^{(+)}_k) &\subset& \left\{
\begin{array}{cc} C^{+-} & s \geq k+3/2 \\ \emptyset & s < k+3/2   \end{array}  \right. \, ,\\
WF^s(\W_\omega^\Gamma - {\,}^\psi \had^{(-)}_k) &\subset& \left\{
\begin{array}{cc} C^{-+} & s \geq k+3/2 \\ \emptyset & s < k+3/2   \end{array}  \right.
\, .
\end{eqnarray}
\end{corollary}

\section{A Sobolev point-splitting result}\label{splitting section}

We quote a result, taken from \cite{fewster verch}, for smooth functions:

\begin{lemma}\label{smooth splitting}
If $f\in C^\infty_0(\R)$ and $u\in C^\infty_0(\R\times\R)$ then the
following identity holds:
\begin{equation}\label{rhs}
\int_\R \rmd t \, f^2(t)u(t,t) = \int_{\R\times\R}
\frac{\rmd\xi\rmd\xi'}{(2\pi)^2}\widehat{f^2}(\xi-\xi')\hat{u}(-\xi,\xi')
\, .
\end{equation}
\end{lemma}

This result forms the basis of the analysis in \cite{fewster verch, dawson fewster} where the authors relate $u$ to the energy density obtained from normal ordering, i.e. 
\begin{equation}\label{normal ordered energy}
\frac{1}{2}\bigg( \mathbbm{1} \otimes D - D\otimes\mathbbm{1}  \bigg)\vartheta^*\big( \W_\omega-\W_{\omega_0} \big)(t,t)
\end{equation}
 where $\omega_0$ is another Hadamard state of the Dirac field. Since the difference between any two two-point functions arising from Hadamard states is smooth the quantity (\ref{normal ordered energy}) is readily identifiable with $u$ in the hypothesis of lemma \ref{smooth splitting}. However, $\W_\omega - {\,}^\psi \had_1^{(+)}$ featuring in (\ref{Had regularised energy}) is not smooth so we need to relax the hypothesis of lemma \ref{smooth splitting} in order to proceed. In particular we wish to show that one has the conclusion of lemma \ref{smooth splitting} under the weaker assumption that $u\in H^{s}(\R\times\R)\cap \E'(\R\times\R)$ for $s>1$. Let $u\in C^\infty_0(\R\times\R)$; then, applying the H\"older inequality we have
\begin{eqnarray}
\int_\R \rmd t \, \big| f^2(t)u(t,t)  \big| &\leq& \norm f^2
\norm_{L^1(\R)} \, \sup_{t\in\R}|u(t,t)| \\
&\leq& \norm f^2 \norm_{L^1(\R)} \, \norm u \norm_{L^\infty(\R\times\R)} 
\end{eqnarray}
for all $u\in C^\infty_0(\R\times\R)$; we also remark that $\norm \cdot\norm_{L^\infty(\R\times\R)}$ is the natural norm on $C(\R\times\R)$. Since the embedding of
$H^{s}(\R\times\R)$ into $C(\R\times\R)$ is continuous for $s>1$ there
exists a constant $c>0$ such that $\norm u
\norm_{L^\infty(\R\times\R)}\leq c \norm
u\norm_{H^{s}(\R\times\R)}$ for all $u\in
H^{s}(\R\times\R)$ and $s>1$. Hence, 
\begin{equation}\label{BLT2}
\int_\R \rmd t \, |f^2(t)u(t,t)| \leq c\norm f^2
\norm_{L^1(\R)} \, \norm u \norm_{H^{s}(\R\times\R)}  \quad \forall u \in C^\infty_0(\R\times\R) \quad s>1 \, .
\end{equation}
Moreover, as $C^\infty_0(\R\times\R)$ is dense in $H^{s}(\R\times\R)$, the Bounded Linear Transform theorem implies that (\ref{BLT2}) holds for all $u\in H^{s}(\R\times\R)$ for $s>1$. Equally, we may apply the H\"older inequality to the right-hand-side of (\ref{rhs}) to obtain
\begin{equation}
\bigg| \int_{\R\times\R}\frac{\rmd\xi\rmd\xi'}{(2\pi)^2} \widehat{f^2}(\xi-\xi')\widehat{u}(-\xi,\xi') \bigg| \leq \norm F \norm_{L^\infty(\R\times\R)} \, \norm \widehat{u}\norm_{L^1(\R\times\R)} \, ,
\end{equation}
where $F(\xi,\xi')=\widehat{f^2}(\xi-\xi')$. A factor of
$(1+|\xi|^2+|\xi'|^2)^{(1+\varepsilon)/2}(1+|\xi|^2+|\xi'|^2)^{-(1+\varepsilon)/2}$
is then introduced into the $L^1$ norm and the Cauchy-Schwarz inequality
applied to obtain
\begin{eqnarray}\label{BLT1}
\norm \widehat{u}\norm_{L^1(\R\times\R)} &\leq& \sqrt{\frac{\pi}{\epsilon}} \norm u \norm_{H^{1+\varepsilon}(\R\times\R)} \, 
\end{eqnarray}
where we have written $s>1$ as $s=1+\epsilon$. Again the Bounded Linear Transform theorem implies that (\ref{BLT1}) holds for all $u\in H^{1+\epsilon}(\R\times\R)$ where the $\,\widehat{\,}\,$ now refers to the continuous extension of the Fourier transform to $H^{1+\epsilon}(\R\times\R)$ (although this must agree with the usual Fourier transform on $L^2(\R\times\R)$ or $\s'(\R\times\R)$). Therefore, since (\ref{rhs}) holds for a dense subset of $H^{1+\epsilon}(\R\times\R)$ and may be extended continuously onto the whole of the space, we have proven:

\begin{lemma}\label{dirac point split 1}
Let $f\in C^\infty_0(\R)$ and $u\in H^s(\R\times\R)\cap\E'(\R\times\R)$,
$s >1$, then the following identity holds:
\begin{equation}
\int_\R \rmd t \, f^2(t) u(t,t) = \int_{\R\times\R}
\frac{\rmd\xi\rmd\xi'}{(2\pi)^2}
\widehat{f^2}(\xi-\xi')\hat{u}(-\xi,\xi') \, .
\end{equation}
\end{lemma}

It is clear under the identification $u(t,t')=\langle \rho^\fin\rangle_\omega(\gamma(t))$ (cf. (\ref{Had regularised energy})) that
\begin{eqnarray}
&\,& \int_\R \rmd t \, f^2(t)\langle \rho^\fin\rangle_\omega (\gamma(t)) \nnnl &\,& \qquad = \frac{1}{2}\int_{\R\times\R}\frac{\rmd\xi \rmd\xi'}{(2\pi)^2} (\xi+\xi')\widehat{f^2}(\xi-\xi') \big[\vartheta^*( \W_\omega - {\,}^\psi \had^{(+)}_1  )\big]^\wedge(-\xi,\xi') \label{dirac int} \, .
\end{eqnarray}

We now prove a result (similar to lemma 5 of \cite{dawson fewster} in all but one detail of the proof\footnote{The distinct step is contained within line (\ref{distinction}).}) which will enable us to relate the right hand side of (\ref{dirac int}) to an integral over the diagonal:

\begin{lemma}\label{dirac point split 2}
If $f\in C^\infty_0(\R)$ is real valued and $u\in H^s(\R\times\R)$, $s>2$, is compactly
supported then 
\begin{eqnarray}
\int_{\R\times\R}\frac{\rmd\xi\,\rmd\xi'}{(2\pi)^2} \,
(\xi+\xi')\widehat{f^2}(\xi-\xi')\hat{u}(-\xi,\xi') =
\frac{1}{\pi}\int_\R \rmd \xi \, \xi \, u(\overline{f^\xi},f^\xi)
\end{eqnarray}
where $f^\xi(t)=e^{i\xi t}f(t)$. 
\end{lemma}
\begin{proof} Lemma 6.1 of \cite{fewster verch} states that
\begin{equation}
(\xi+\xi')\widehat{f^2}(\xi-\xi')=\frac{1}{\pi}\int_\R \rmd\zeta\,
\zeta\hat{f}(\xi-\zeta)\overline{\hat{f}(\xi'-\zeta)}
\end{equation}
and therefore
\begin{eqnarray}
&\,& \int_{\R\times\R}\frac{\rmd\xi\,\rmd\xi'}{(2\pi)^2}(\xi+\xi')\widehat{f^2}(\xi-\xi')\hat{u}(-\xi,\xi') \nnnl
&\,& \quad =
\frac{1}{\pi}\int_{\R\times\R}\frac{\rmd\xi\,\rmd\xi'}{(2\pi)^2}\int_\R\rmd\zeta
\, \zeta
\hat{f}(\xi-\zeta)\overline{\hat{f}(\xi'-\zeta)}\hat{u}(-\xi,\xi') \,
. \label{thingy}
\end{eqnarray}
 We also note, by a simple application of the convolution theorem, that 
\begin{eqnarray}
u(\overline{f^\xi},f^{\xi}) &=& [(f\otimes
f)u](e^{-i\xi\cdot},e^{i\xi\cdot}) \\
&=&
\int_{\R\times\R}\frac{\rmd\zeta\,\rmd\zeta'}{(2\pi)^2}\hat{f}(-\xi-\zeta)\hat{f}(\xi-\zeta')\hat{u}(\zeta,\zeta')
\\
&=& \int_{\R\times\R}\frac{\rmd\zeta \,
\rmd\zeta'}{(2\pi)^2}\hat{f}(-\xi+\zeta)\hat{f}(\xi-\zeta')\hat{u}(-\zeta,\zeta')
\\
&=&
\int_{\R\times\R}\frac{\rmd\zeta\,\rmd\zeta'}{(2\pi)^2}\,\hat{f}(\zeta-\xi)\overline{\hat{f}(\zeta'-\xi)}
\, \hat{u}(-\zeta,\zeta')
\end{eqnarray}
and, therefore, that the statement of the theorem will be established if
the integrals in (\ref{thingy}) can be reordered. If we estimate $|\hat{f}(x)|\leq
c/(1+|x|^2)$ then, by the arithmetic-geometric mean
inequality
\begin{eqnarray}
&\,& \int_\R \rmd\zeta \, \big| \zeta \hat{f}(\xi-\zeta)\overline{\hat{f}(\xi'-\zeta)} 
\big| \nnnl
&\,& \qquad \leq \frac{c^2}{2}\int_\R \rmd\zeta \,\bigg(
\frac{|\zeta|}{(1+|\xi-\zeta|^2)^2}+\frac{|\zeta|}{(1+|\xi'-\zeta|^2)^2} \bigg)\\
&\,& \qquad \leq \frac{c^2}{2}(2+|\xi\arctan\xi|+|\xi'\arctan\xi'|) \\
&\,& \qquad \leq \frac{c^2\pi}{4}(2+|\xi|+|\xi'|) \, .
\end{eqnarray}
Then
\begin{eqnarray}
&\,& \frac{1}{\pi}\int_{\R\times\R\times\R}
\frac{\rmd\xi\,\rmd\xi'}{(2\pi)^2} \, \rmd\zeta \, \bigg| \zeta \,
\hat{f}(\xi-\zeta)\overline{\hat{f}(\xi'-\zeta)}\hat{u}(-\xi,\xi')
\bigg| \nnnl 
&\,& \quad \leq \frac{c^2}{4}\int_{\R\times\R}
\frac{\rmd\xi\,\rmd\xi'}{(2\pi)^2} (2+|\xi|+|\xi'|) |\hat{u}(-\xi,\xi')|
\\
&\,& \quad \leq \frac{c^2}{16\pi^2}\norm p_s \norm_{L^2(\R\times\R)} \, \norm u \norm_{H^s(\R\times\R)} \label{distinction}
\end{eqnarray}
where we have used the Cauchy-Schwarz inequality and written
\begin{equation}
p_s(\xi,\xi')=\frac{2+|\xi|+|\xi'|}{(1+|\xi|^2+|\xi'|^2)^{s/2}} \, .
\end{equation}
The $L^1$ norm of $\zeta\hat{f}(\xi-\zeta)\overline{\hat{f}(\xi'-\zeta)}\hat{u}(-\xi,\xi')$ will be finite if $u\in H^s(\R\times\R)$ and $s>2$.
Under these conditions, Fubini's
theorem implies that the integrals can be reordered to obtain the
desired result.
\end{proof}

\section{A worldline absolute quantum weak energy inequality}\label{punch line}

We are now in a position to state our result concerning the Dirac field.

\begin{theorem}\label{main result}
Let $\omega$ be a Hadamard state for the Dirac field, $\gamma$ be a timelike worldline and $f\in C^\infty_0(\R)$ be real valued; then
\begin{equation}
\int_\R \, \rmd t \, f^2(t) \langle \rho^\fin \rangle_\omega (\gamma(t)) \geq - B
\end{equation}
where 
\begin{eqnarray}
B &=& \int_{\R^+}\frac{\rmd\xi}{2\pi}\xi \bigg[ f\otimes f \, \vartheta^*{\,}^\psi\had^{(+)}_4 
\bigg]^\wedge(-\xi,\xi) \nnnl
&\,& \quad - \int_{\R^-}\frac{\rmd\xi}{2\pi} \xi \bigg[ f\otimes f
\,\vartheta^*\big(i\mathrm{S}_\spin-{\,}^\psi\had^{(+)}_4 \big) \bigg]^\wedge(-\xi,\xi)
\end{eqnarray}
and $\vartheta =\gamma\otimes\gamma$.
\end{theorem}
\begin{proof}
Corollary \ref{dirac grading} implies that $\W_\omega-{\,}^\psi \had_4^{(+)}
\in H^{5+\varepsilon}_\loc(X)$ from which it follows that
\begin{equation}\label{dirac suff reg}
f\otimes f\vartheta^*\bigg(\big(\mathbbm{1}\otimes ie_0\cdot\nabla
-ie_o\cdot\nabla \otimes \mathbbm{1}  \big) \big( \W_\omega-{\,}^\psi \had^{(+)}_4
\big)\bigg) \in H^{1+\varepsilon}_\loc(\R\times\R) \, ,
\end{equation}
where we have lost one Sobolev order as a result of differentiation and
a further three from the restriction to $\R\times\R$. Lemma \ref{dirac point split 1}, and the remarks following the proof, enable us to write
\begin{eqnarray}
&\,& \int_\R \rmd t \, f^2(t)\langle  \rho^\fin \rangle_\omega (\gamma(t)) \nnnl&\,& \qquad = \frac{1}{2}\int_{\R\times\R}\frac{\rmd\xi\rmd\xi'}{(2\pi)^2}(\xi+\xi')\widehat{f^2}(\xi-\xi')\big[ \vartheta^*(\W_\omega - {\,}^\psi\had^{(+)}_4)  \big]^\wedge (-\xi,\xi') \, .
\end{eqnarray}
and as $\vartheta^*(\W_\omega-{\,}^\psi\had^{(+)}_4)\in H^{2+\varepsilon}_\loc(\R\times\R)$ we can employ lemma \ref{dirac point split 2} to obtain:
\begin{equation}
\int_\R \rmd t \, f^2(t)\langle \rho^\fin\rangle_\omega(t) = \int_\R \frac{\rmd\xi}{2\pi}\xi \vartheta^*(\W_\omega-{\,}^\psi\had^{(+)}_4)(\overline{f^\xi},f^\xi) \, .
\end{equation}
We decompose the integral into its positive and negative frequency components and
appeal to the anticommutation relation (in scalar form) $\W_\omega+\W_\omega^\Gamma = i\mathrm{S}_\spin$
to write
\begin{eqnarray}
\int_\R \rmd t \, f^2(t) \langle \rho^\fin\rangle_\omega(t) &=& \int_{\R^+}\frac{\rmd\xi}{2\pi}\xi\vartheta^*\big(\W_\omega-{\,}^\psi\had^{(+)}_4\big)(\overline{f^\xi},f^\xi) \nnnl
&\,& \qquad + \int_{\R^-}\frac{\rmd\xi}{2\pi}\xi\vartheta^*\big(  i\mathrm{S}_\spin  -\W_\omega^\Gamma - {\,}^\psi\had^{(+)}_4 \big)(\overline{f^\xi},f^\xi) \, .
\end{eqnarray}
Recall that $\W_\omega$ and $\W^\Gamma_\omega$ are distributions of positive type, hence
\begin{eqnarray}
\int_\R \rmd t \, f^2(t) \langle \rho^\fin\rangle_\omega(t) &\geq& -\int_{\R^+}\frac{\rmd\xi}{2\pi}\xi\vartheta^*{\,}^\psi\had^{(+)}_4(\overline{f^\xi},f^\xi) \nnnl
&\,& \qquad + \int_{\R^-}\frac{\rmd\xi}{2\pi}\xi\vartheta^*\big(  i\mathrm{S}_\spin - {\,}^\psi\had^{(+)}_4 \big)(\overline{f^\xi},f^\xi) \, 
\end{eqnarray}
and it remains to show that this lower bound is finite. 

We make the replacements of ${\,}^\psi\had^{(+)}_4=\W_{\omega_0}-F$ and $i\mathrm{S}_\spin-{\,}^\psi\had^{(-)}_4=\W^\Gamma_{\omega_0}-G$ where $F,G\in C^4(X)$ and $\W_{\omega_0},\W^\Gamma_{\omega_0}$ arise from some arbitrary Hadamard
state $\omega_0$. We remark that this replacement is a technical device only which
we introduce to prove finiteness: \emph{The bound is still independent of any state}. Hence, we have
\begin{eqnarray}
B &=& \int_{\R^+}\frac{\rmd\xi}{2\pi} \xi \bigg[ f\otimes f \, \vartheta^*\W_{\omega_0}  \bigg]^\wedge(-\xi,\xi) -\int_{\R^-}\frac{\rmd\xi}{2\pi}\xi \bigg[ f\otimes f \, \vartheta^*\W^\Gamma_{\omega_0} \bigg]^\wedge(-\xi,\xi) \nnnl
&\,& \quad -\int_{\R^+}\frac{\rmd\xi}{2\pi}\xi\bigg[ f\otimes f \vartheta^*F \bigg]^\wedge(-\xi,\xi) + \int_{\R^-}\frac{\rmd\xi}{2\pi}\xi\bigg[ f\otimes f \vartheta^*G \bigg]^\wedge(-\xi,\xi) \, .
\end{eqnarray}
The finiteness of the $\W_{\omega_0}$ and $\W^\Gamma_{\omega_0}$ pieces is proven by the wave-front set conditions, $WF(\W_{\omega_0})\subset C^{+-}$ and $WF(\W^\Gamma_{\omega_0})\subset C^{-+}$, which imply\footnote{For the full details of this step the reader is directed to \S 3.2 of \cite{dawson fewster}.} that $\vartheta^*\W_{\omega_0}$ and $\vartheta^*\W_{\omega_0}^\Gamma$ are rapidly decaying in the directions they are being integrated in.  Finally, one may use the following estimates 
\begin{eqnarray}
[f\otimes f \, \vartheta^*F ]^\wedge(\xi,\xi') &\leq& \frac{c}{(1+|\xi|^2+|\xi'|^2)^2} \, , \\
\left[f\otimes f \, \vartheta^*G \right]^\wedge(\xi,\xi') &\leq& \frac{c'}{(1+|\xi|^2+|\xi'|^2)^2} \, .
\end{eqnarray}
Hence, our bound is finite.
\end{proof} 

Theorem \ref{main result} enables us to finally formulate our absolute QWEI for the Dirac field. Wald's uniqueness theorem implies that the regularised energy density $\langle \rho^\fin \rangle_\omega$ we have computed is equal to the renormalised energy density $\langle \rho^\textrm{ren}\rangle_\omega$ up to the addition of a local curvature term $C$. Hence, our result reads:
\begin{equation}
\int_\R \rmd t \, f^2(t) \, \langle \rho^\textrm{ren}\rangle_\omega (t) \geq - \mathcal{B}
\end{equation}
where $\mathcal{B}$ is given by 
\begin{eqnarray}
\mathcal{B} &=& \int_{\R^+}\frac{\rmd\xi}{2\pi}\xi \bigg[ f\otimes f \, \vartheta^*{\,}^\psi\had^{(+)}_4 
\bigg]^\wedge(-\xi,\xi) \nnnl
&\,& \quad - \int_{\R^-}\frac{\rmd\xi}{2\pi} \xi \bigg[ f\otimes f
\,\vartheta^*\big(i\mathrm{S}_\spin-{\,}^\psi\had^{(+)}_4 \big) \bigg]^\wedge(-\xi,\xi) \nnnl
&\,& \qquad - \int_\R \, \rmd t \, f^2(t)C(t) 
\end{eqnarray}

As reported in \cite{fewster smith}, where a more complete discussion of the renormalisation of the stress tensor of a quantum field may be found, the view may be held that the value of this curvature term (alongside the mass and curvature coupling) is an essential detail in the specification of the theory and that $C$ should be, at least in principle, measurable. Alternatively, one may hold the view that this unavoidable ambiguity is a manifestation of a breakdown of the semi-classical theory and that a more complete theory of quantum gravity is needed.

\section{Conclusion}

We have succeeded in proving a new absolute QWEI for the Dirac
field under general circumstances. By exploiting a Sobolev graded refinement of H\"ormander's wave-front set we have been able to modify the proof Fewster and Dawson
\cite{dawson fewster} give for their difference QWEI and remove any
reference to a state in the bound. Moreover, it is straightforward to
use the techniques of \cite{fewster smith} to obtain additional $WF^s$
information of the constituents of the Dirac Hadamard series ${\,}^\psi H^{(\pm)}_k$.

\begin{center}\textbf{Acknowledgements}\end{center}

The author would like to thank C.J. Fewster for his guidance over the course of this research. Additional thanks go to J.A. Sanders, L.W. Osterbrink, S.P. Dawson and P.~Watts for their helpful comments on the manuscript.

\end{document}